\documentclass[11pt]{article}

\usepackage{latexsym}
\usepackage{theorem}
\usepackage{color,graphicx}
\usepackage{amssymb}
\usepackage{amsmath}
\usepackage{txfonts}
\usepackage{enumitem}
\usepackage{fullpage}

\newenvironment{proof}[1][Proof]
{\par\noindent{\bf #1:} }{\hspace*{\fill}\nolinebreak{$\Box$}\bigskip\par}
\newcommand{\qed}{\hspace*{\fill}\nolinebreak\ensuremath{\Box}}

\newtheorem{theorem}{Theorem}
\newtheorem{lemma}{Lemma}[section]
\newtheorem{claim}[lemma]{Claim}
\newtheorem{corollary}[theorem]{Corollary}
\newtheorem{observation}[lemma]{Observation}

\theorembodyfont{\upshape}

\newcommand{\cS}{\mathcal{S}}
\newcommand{\cP}{\mathcal{P}}

\newcommand{\cM}{\mathcal{M}}
\newcommand{\cI}{\mathcal{I}}

\newcommand{\Mshared}[1][]{\cM \if!#1!\else (#1) \fi}
\newcommand{\Mpriv}{\cP}

\newcommand{\complTime}[3]{C_{#1}^{#3}(#2)}  %{C_{#2}(#1,#3)}            % #1=schedule, #2=job, #3=processor
\newcommand{\startTime}[3]{s_{#1}^{#3}(#2)}  %{s_{#2}(#1,#3)}
\newcommand{\tct}[1]{\varSigma(#1)}                  % #1=schedule

\newcommand{\reals}{\mathbb{R}}
\newcommand{\st}{\hspace{0.1cm}\bigl|\bigr.\hspace{0.1cm}}

\newcommand{\jobs}{\mathcal{J}}

\newcommand{\probShort}{\textup{\texttt{WSPS}}}

\renewcommand\footnotemark{}
\begin{document}

\title{\textbf{Shared processor scheduling}}

\author{
  Dariusz Dereniowski\\
  \small{\emph{Faculty of Electronics,}}\\
  \small{\emph{Telecommunications and Informatics},}\\
  \small{\emph{Gda{\'n}sk University of Technology},}\\
  \small{\emph{Gda{\'n}sk, Poland}}
\and
  Wies{\l}aw Kubiak\\
  \small{\emph{Faculty of Business Administration},}\\
  \small{\emph{Memorial University},}\\
  \small{\emph{St. John's, Canada}}
}

\date{\let\thefootnote\relax\footnote{Emails: deren@eti.pg.gda.pl (Dariusz Dereniowski) and wkubiak@mun.ca (Wies{\l}aw Kubiak)}}

\maketitle

\begin{abstract}
We study the shared processor scheduling problem with a single shared processor where a unit time saving (weight) obtained by processing a job on the shared processor depends on the job.  A polynomial-time optimization algorithm has been given for the problem with equal weights in the literature. This paper extends that result by showing an $O(n \log n)$ optimization algorithm for a class of instances in which non-decreasing order of jobs with respect to processing times provides a non-increasing order with respect to weights --- this instance generalizes the unweighted case of the problem. This algorithm also leads to a $\frac{1}{2}$-approximation algorithm for the general weighted problem. The complexity of the weighted problem remains open.
\end{abstract}


\section{Introduction}
Consider a subcontracting system in which each agent $j$ has a job of duration $p_j$ to be executed.
Such an agent can perform the work by itself, in which case the job ends after $p_j$ units of time, or it can send (subcontract) a part of length $s_j\leq p_j/2$ of this job to a subcontractor for processing.
The subcontractor needs to complete this piece of agent's $j$ job by $p_j-s_j$, i.e., the speedup in terms of the completion time that the agent achieves in this scenario is exactly $s_j$, or in other words, the work of agent $j$ is completed at time moment $p_j-s_j$.
Whenever $s_j>0$, the subcontractor is rewarded by agent $j$: the payoff of executing $s_j$ units of $j$-th agent's job is $s_jw_j$.
The goal of the subcontractor is to maximize its total payoff under the condition that the parts of jobs received from different agents cannot be executed simultaneously by the subcontractor.
Thus, in this subcontracting system all agents try to minimize completion times of their jobs (the parameters $p_j$ and $w_j$ are fixed for each agent $j$), i.e. they are willing to commission the biggest possible part of their jobs to the subcontractor.
The subcontractor is the party that decides what amount, if any, of each job to process in order to maximize its total payoff.

The shared processor scheduling problem can be placed in a wider context of scheduling with presence of private (local) processors (machines), available only to a particular job or a set of jobs, and shared (global) processors that are available to all jobs.
Then, some additional rules are given in order to specify the conditions under which a job can gain access to a shared processor in such systems.
These systems can be run as either centralized or decentralized.
The former typically has a single optimization criterion forcing all parties to achieve the same goal.
The latter emphasizes that each party is trying to optimize its own goal, which may (and often does) lead to problems having no solutions which are optimal for each agent (job) individually.
These problems can be seen as multi-criteria optimization or coordination problems. The latter and can be further subdivided into problems in which agents have complete knowledge about resources of other agents (complete information games) and problems without such a full knowledge (distributed systems) in the search for coordinating mechanisms.
This work falls into the category of centralized problems as the subcontractor is deciding on the schedule that reflects its best interest.

\medskip
The outline of this paper is as follows.
In the next section we briefly survey the related work to provide a state of the art overview.
Section~\ref{sec:problem} gives a formal statement of the scheduling problem we study and it introduces the necessary notation.
Then, in Section~\ref{sec:preliminaries}, we recall some facts related to the problem, mainly the fact that when computing optimal schedules one may restrict attention to schedules that are called \emph{synchronized}.
This generally greatly reduces the formal arguments and algorithmic approach.
Section~\ref{sec:antithetical} considers a restricted version of the problem in which it is assumed that for any pair of jobs, neither of the jobs can have weight and processing time to be strictly smaller than the other.
We give an $O(n\log n)$-time optimization algorithm for this case, and we use it subsequently as a building block to obtain an $O(n\log n)$-time $1/2$-approximation algorithm for the general case in Section~\ref{sec:approx}.

\section{Related work} \label{sec:related}

The shared processor scheduling problem has recently been studied by  Vairaktarakis and Aydinliyim  \cite{VairaktarakisAydinliyim07},  Hezarkhani and Kubiak \cite{HK15}, and  Dereniowski and Kubiak \cite{DK16}.
Vairaktarakis and Aydinliyim  \cite{VairaktarakisAydinliyim07} consider the (unweighted) problem with a single shared  processor and with each job allowed to use at most  one time interval on the shared  processor. This case is sometimes referred to as non-preemptive since jobs are not allowed preemption on the shared processor.  \cite{VairaktarakisAydinliyim07} proves that there are optimal schedules that complete job execution on private and shared processor at the same time, we call such schedules \emph{synchronized}, for the non-preemptive case with equal weights. It further shows that this guarantees that sequencing jobs in ascending order of their processing times leads to an optimal solution for the case.  \cite{HK15} observes that this algorithm also gives optimal solutions to the preemptive problem, 
where more than one interval can be used by a job on the shared processor, provided that all weights are equal. \cite{DK16} consider shared multi-processor problem proving its strong NP-hardness and giving an efficient, polynomial-time algorithm 
for the shared multi-processor problem with equal weights.
Also, it is shown in \cite{DK16} that synchronized optimal schedules always exist for weighted multi-processor instances.
Vairaktarakis and Aydinliyim \cite{VairaktarakisAydinliyim07} , Vairaktarakis  \cite{V13}, and Hezarkhani and Kubiak \cite{HK15} also study decentralized subcontracting systems focusing on coordinating mechanisms to ensure their efficiency.

The motivation to study the shared processor scheduling problem comes from diverse applications. Vairaktarakis and Aydinliyim  \cite{VairaktarakisAydinliyim07} consider it in the context of  supply chains were subcontracting allows jobs to reduce their completion times by using a shared subcontractor's processor. Bharadwaj et. al.  \cite{HBGR03}  use the divisible load scheduling to reduce a job completion time in parallel and distributed computer systems, and  Anderson \cite {A81} argues for using
batches of potentially infinitely small items that can be processed independently of other items of the batch in scheduling job-shops. We refer the reader to Dereniowski and Kubiak \cite{DK16} for more details on these applications.

We also remark multi-agent scheduling models in which each agent has its own optimality criterion and performs actions aimed at optimizing it.
In these models, being examples of decentralized systems, agents usually have a number of non-divisible jobs to execute (depending on the optimization criterion this may be seen as having one divisible job, but restricted by allowing preemptions only at certain specified points).
For minimization of weighted total completion time in such models see Lee et. al.~\cite{LeeCLP09} and weighted number of tardy jobs see Cheng, Ng and Yuan~\cite{ChengNY06}.
Bukchin and Hanany~\cite{BukchinH07} give an example of a game-theoretic analysis to a problem of this type.
For overviews and further references on the multi-agent scheduling we refer to the book by Agnetis et. al. \cite{ABGPS14}.

\section{Problem formulation} \label{sec:problem}
We are given a set $\jobs$ of $n$ preemptive jobs.
Each job $j\in\jobs$ has its processing time $p_j$ and weight $w_j$.
With each job $j\in\jobs$ we associate its \emph{private} processor denoted by $\Mpriv_j$.
Moreover, there exists a single shared processor, denoted by $\Mshared$, that is available for all jobs. We follow the convention and notation from~\cite{DK16} to formulate the problem in this paper.
%The considered problem may have several equivalent formulations and in this work we use the formulation from~\cite{DK16}.

A schedule $\cS$ is \emph{feasible} if satisfies the following conditions:
\begin{itemize}
 \item each job $j\in\jobs$ executes non-preemptively in a \emph{single} time interval $(0,\complTime{\cS}{j}{\Mpriv})$ on its private processor and there is a (possibly empty) collection of open intervals $\cI_j$ such that $j$ executes non-preemptively in each time interval $I\in\cI_j$ on the shared processor,
 \item for each job $j\in\jobs$,
 \[\complTime{\cS}{j}{\Mpriv}+\bigcup_{I\in\cI_j}|I| = p_j,\]
 \item the time intervals in $\bigcup_{j\in\jobs}\cI_j$ are pairwise disjoint.
\end{itemize}

Given a feasible schedule $\cS$, for each job $j\in\jobs$ we call any maximal time interval in which $j$ executes on both private $\Mpriv_j$ and shared $\Mshared$ simultaneously an \emph{overlap}.
The total overlap $t_j$ of job $j$ equals the sum of lengths of all overlaps for $j$.
The \emph{total weighted overlap} of $\cS$ equals
\[\tct{\cS}=\sum_{j\in\jobs}t_jw_j.\] 
A feasible schedule that maximizes the total weighted overlap is called \emph{optimal}.

The formulation our \emph{Weighted Single-Processor Scheduling} problem ($\probShort$), is as follows.
\begin{itemize} [leftmargin=2.0cm]
 \item[Instance:] A set of weighted jobs $\jobs$ with arbitrary given processing times.
 \item[Goal:] Find an optimal schedule for $\jobs$.
\end{itemize}

\section{Preliminaries} \label{sec:preliminaries}

Let $\cS$ be a feasible schedule.
We denote by $\startTime{\cS}{j}{\Mshared}$ and $\complTime{\cS}{j}{\Mshared}$ the start time and the completion times of a job $j$ on the shared processor, respectively.
For brevity we take $\startTime{\cS}{j}{\Mshared}=\complTime{\cS}{j}{\Mshared}=0$ if a job $j$ executes on its private processor only.
Whenever $\startTime{\cS}{j}{\Mshared}<\complTime{\cS}{j}{\Mshared}$, i.e., some non-empty part of a job $j$ executes on $\Mshared$, then we say that the job $j$ \emph{appears on} $\Mshared$ in schedule $\cS$.
If, in a schedule $\cS$, there is no idle time on the shared processor in time interval
\[\left[0,\max\{\complTime{\cS}{j}{\Mshared}\st j\in \jobs\}\right],\]
then we say that $\cS$ \emph{has no gaps}. We have the following results form the literature.
\begin{observation}[\cite{DK16}] \label{obs:idle-time}
There exists an optimal schedule that has no gaps.
\end{observation}
A schedule $\cS$ is called \emph{non-preemptive} if each job $j$ executes in $\cS$
in time interval $[\startTime{\cS}{j}{\Mshared},\complTime{\cS}{j}{\Mshared}]$ on the shared processor.
We say that a schedule is \emph{synchronized} if it satisfies the following conditions:
\begin{enumerate}[label={\normalfont{(\roman*)}}]
 \item it is non-preemptive and has no gaps,
 \item for each job $j$ that appears on the shared processor it holds $\complTime{\cS}{j}{\Mshared}=\complTime{\cS}{j}{\Mpriv}$.
\end{enumerate}

\begin{theorem}[\cite{DK16}] \label{thm:synchronized}
There exists an optimal synchronized schedule.
\end{theorem}
Consider a synchronized schedule $\cS$.
Let $A=\{j_1,\ldots,j_k\}\subseteq\jobs$ be the set of all jobs that appear on $\Mshared$ in $\cS$, where the jobs are ordered according to their completion times in $\cS$, i.e. $\complTime{\cS}{j_1}{\Mshared}<\cdots<\complTime{\cS}{j_k}{\Mshared}$.
Note that the set $A$ and the order are enough to determine the schedule $\cS$.
Indeed, given the order $(j_1,\ldots,j_k)$ we obtain that for each $i\in\{1,\ldots,k\}$ (by proceeding with increasing values of $i$),
\[\startTime{\cS}{j_i}{\Mshared}=\complTime{\cS}{j_{i-1}}{\Mshared} \quad\textup{and}\quad \complTime{\cS}{j_i}{\Mshared}=\complTime{\cS}{j_i}{\Mpriv}=\left(p_i+\startTime{\cS}{j_i}{\Mshared}\right)/2,\]
where $\complTime{\cS}{j_0}{\Mshared}=0$.
This formula implies that the start times and completion times can be iteratively computed for all jobs.
Thus for synchronized schedules we write for brevity $\cS=(j_1,\ldots,j_k)$ to refer to the schedule computed above.

\section{A $O(n log n)$-time optimal algorithm for antithetical instances} \label{sec:antithetical}
We call an instance $\jobs$ of the problem \emph{antithetical} if for any two jobs $i$ and $j$ it holds: $p_i\leq p_j$ implies $w_i \geq w_j$.
We call a schedule $\cS$ \emph{processing time ordered} if $\cS=(j_1,\ldots,j_n)$, where $p_{j_i}\leq p_{j_{i+1}}$ for each $i\in\{1,\ldots,n-1\}$.
Observe that by construction, $\cS$ is synchronized and all jobs from $\jobs$ appear on the shared processor, see \cite{VairaktarakisAydinliyim07} and \cite{HK15}.
We now prove that any processing time ordered schedule is an optimal solution for an antithetical instance.  This gives an $O(n\log n)$-time optimization algorithm for antithetical instances. 
We remark that this algorithm generalizes the previously known solutions for the unweighted case ($w_1=\cdots=w_n$) from \cite{HK15,VairaktarakisAydinliyim07}.

\begin{lemma} \label{lem:antithetical}

A processing time ordered schedule is optimal for any antithetical instance of the problem $\probShort$.
\end{lemma}
\begin{proof}
Let $\cS$ be an optimal schedule for an antithetical instance $\jobs$.
By Theorem~\ref{thm:synchronized} we can assume that $\cS$ is synchronized.
We assume without loss of generality that the jobs in $\jobs=\{1,\ldots,n\}$ are ordered in non-decreasing order of their processing times, i.e. $p_1\leq p_2 \leq \cdots\leq p_n$.
Let $A \subseteq \jobs$ be the set of jobs that appear on the shared processor in $\cS$.
Let $\pi(1), \pi(2), \ldots,\pi(k)$, $k=|A|$, be the order of jobs on the shared processor in $\cS$, i.e., $\cS=(\pi(1), \pi(2), \ldots,\pi(k))$.

The largest index $i\in\{1,\ldots,k-1\}$ such that either
\begin{equation} \label{eq:contradictp}
p_{\pi(i)}>p_{\pi(i+1)}
\end{equation}
or
\begin{equation} \label{eq:contradictA}
\textup{there exists }\bar{j}\in\jobs\setminus A \textup{ such that } \bar{j}>\pi(i+1)
\end{equation}
is called the \emph{violation point of} $\cS$; set the violation point to be zero if no such index $i$ exists.
Note that violation point equals zero only for processing time ordered schedules.

\medskip
Among all optimal and synchronized schedules we take $\cS$ to satisfy the following:
\begin{enumerate} [label={\normalfont{(\alph*)}}]
 \item\label{it:antiA1} the number of jobs that appear on $\Mshared$ in $\cS$ is maximum, and
 \item\label{it:antiA2} with respect to~\ref{it:antiA1}: the violation point of $\cS$ is minimum.
\end{enumerate}

We aim at proving that the violation point of $\cS$ is zero, that is:
\begin{equation} \label{eq:AisJ}
A=\jobs
\end{equation}
and
\begin{equation} \label{eq:pi}
p_{\pi(i)}\leq p_{\pi(i+1)}\textup{ for each }i\in\{1,\ldots,k-1\},
\end{equation}
which immediately implies the lemma.

We prove these claims by contradiction.
Let $i>0$ be the violation point of $\cS$.
By definition, we have that one of the cases \eqref{eq:contradictp} or \eqref{eq:contradictA} holds.
We should arrive at a contradiction in both cases and we start by analyzing the case of \eqref{eq:contradictp}, that is, we assume that \eqref{eq:contradictp} holds for the violation point $i$.
For antithetical instances we have $w_{\pi(i)}\leq w_{\pi(i+1)}$.
Also, for convenience let without loss of generality we denote $j=\pi(i+1)$ and $j+1=\pi(i)$. Thus $p_{j+1}>p_j$ by  (\ref{eq:contradictp}).
In the next two paragraphs we describe a transition from $\cS$ to a new schedule $\cS'$. This transition is depicted in Figure~\ref{fig:exchange}.
\begin{figure}[htb]
\begin{center}
\includegraphics[scale=0.8]{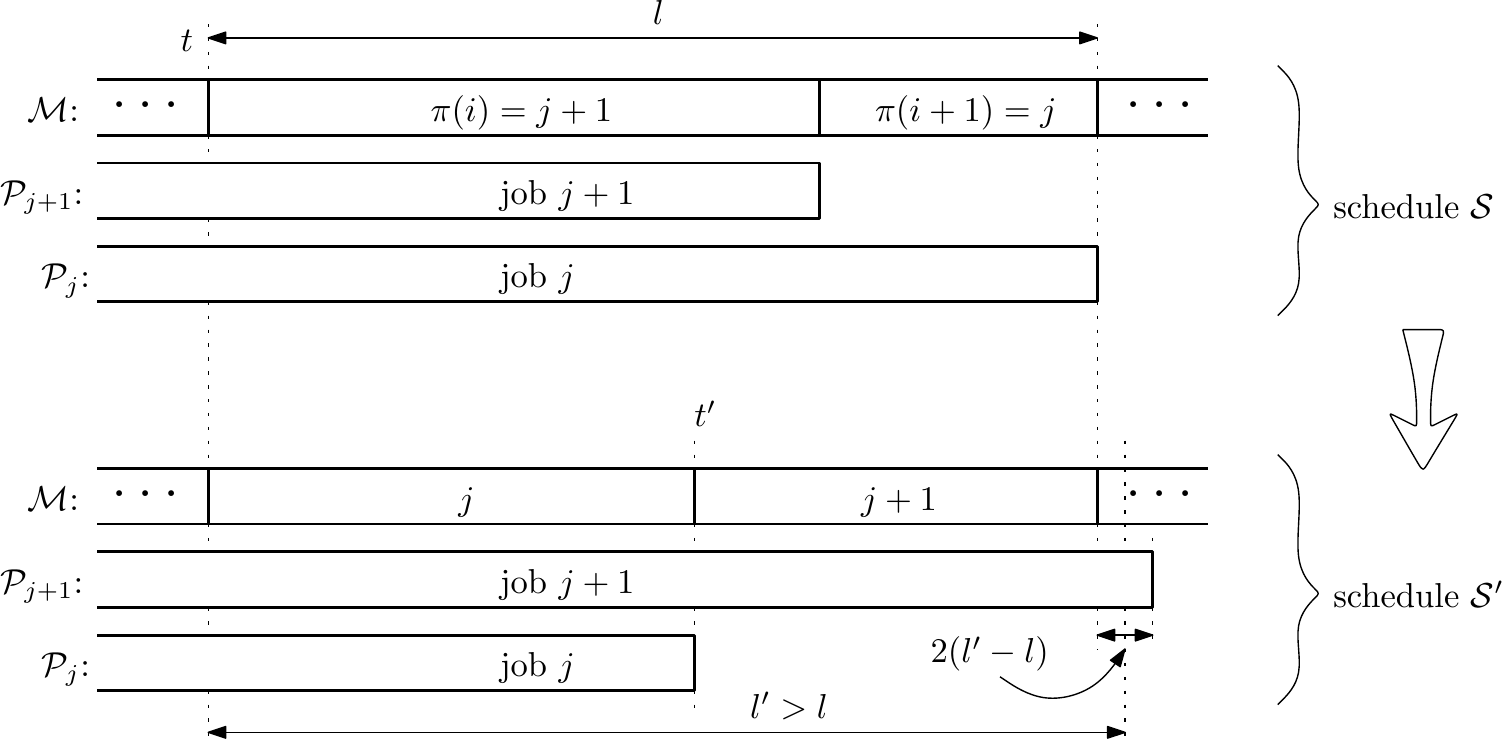}
\caption{Transition from $\cS$ to $\cS'$ in proof of Lemma~\ref{lem:antithetical}}
\label{fig:exchange}
\end{center}
\end{figure}

Consider the intervals in which the two jobs $j$ and $j+1$ execute on $\Mshared$ in $\cS$ and suppose that $j+1$ starts at $t$ on $\Mshared$, $\startTime{\cS}{j+1}{\Mshared}=t$.
Since $\cS$ is synchronized, the length $l=\complTime{\cS}{j}{\Mshared}-\startTime{\cS}{j+1}{\Mshared}$ of this sequence on the shared processor equals
\begin{equation} \label{eq:ldef}
l= \left(\complTime{\cS}{j}{\Mshared} - \startTime{\cS}{j}{\Mshared}\right) + \left(\complTime{\cS}{j+1}{\Mshared} - \startTime{\cS}{j+1}{\Mshared}\right) = \frac{p_{j}-t}{2}+\frac{p_{j+1}-t}{4},
\end{equation}
and its contribution $x$ to the value of objective function (total weighted overlap) equals
\begin{equation}\label{w3}
x=\left(\complTime{\cS}{j+1}{\Mshared} - \startTime{\cS}{j+1}{\Mshared}\right)w_{j+1}+\left(\complTime{\cS}{j}{\Mshared} - \startTime{\cS}{j}{\Mshared}\right)w_j=\frac{p_{j+1}-t}{2}\left(w_{j+1}-\frac{w_j}{2}\right)+\frac{p_{j}-t}{2}w_{j}.
\end{equation}
Thus, we can express the total weighted overlap of $\cS$ as follows:
\begin{equation} \label{eq:tctS}
\tct{\cS}=c+x \quad\textup{for some } c\in\reals.
\end{equation}

Before we formally define $\cS'$, we analyze the impact the reversed order of the two jobs $j$ and $j+1$ in the interval $(t,t+l)$ has on $\cS$ and its objective function. Suppose for the time being that
$j$ starts at $t$ on $\Mshared$ and is followed by the job $j+1$ and that either job is executed in such a way that it  completes on both $\Mshared$ and its private processors at the same time.  Then the length $l'$ of the interval $(t,t+l')$ occupied by these two jobs on the shared processor equals
$$l'=\frac{p_{j+1}-t}{2}+\frac{p_j-t}{4}$$
and its contribution $x'$ to the value of objective function equals
\begin{equation} \label{w1}
x'=\frac{p_j-t}{2}\left(w_j-\frac{w_{j+1}}{2}\right)+\frac{p_{j+1}-t}{2}w_{j+1}.
\end{equation}

Clearly, $l'>l$ for $p_{j+1}>p_j$ by
\begin{equation} \label{eq:lldiff}
l'-l=\frac{p_{j+1}-p_{j}}{4}.
\end{equation}
The job $j$ completes at
\begin{equation} \label{eq:tprimedef}
t'=t+\frac{p_j-t}{2}<t+\frac{p_j-t}{2}+\frac{p_{j+1}-t}{4}=t+l
\end{equation}
after the exchange. We construct the schedule $\cS'$ as follows: $\cS$ and $\cS'$ are identical in time intervals $[0,t)$ and $(t+l,+\infty)$,
the job $j$ executes in time interval
\begin{equation} \label{eq:ttinterval}
(t,t')=(t,t+(p_j-t)/2)
\end{equation}
in $\cS'$ and the job $j+1$ executes in time interval
\[(t',t+l)=(t+(p_j-t)/2,t+l)\]
on processor $\Mshared$ in $\cS'$.
Note that $j$ finishes at the same time on its private and shared processor in $\cS'$ while the job $j+1$ does not have this property. Since $p_{j+1}>p_j$, $j+1$ completes $(p_{j+1}-p_{j})/4$ units later on its private processor. Thus $\cS'$ is not synchronized; see also Figure~\ref{fig:exchange}.
The total weighted overlap of $\cS'$ is then by~\eqref{eq:ldef}, \eqref{eq:tprimedef} and~\eqref{eq:ttinterval}:
\begin{eqnarray} \label{eq:tctSprime}
\begin{aligned}
\tct{\cS'} & = & c + (t'-t)w_j+(t+l-t')w_{j+1}, \\
           & = & c + \frac{p_j-t}{2}w_j + \frac{p_{j+1}-t}{4}w_{j+1}
\end{aligned}
\end{eqnarray}
where $c$ is defined in \eqref{eq:tctS}.
By \eqref{eq:tctS} and \eqref{eq:tctSprime} we obtain that the difference (in total weighted overlaps) between $\cS'$ and $\cS$ is
\begin{equation}
\tct{\cS'}-\tct{\cS} = \frac{p_{j+1}-t}{4}(w_j-w_{j+1}).
\end{equation}
Since $p_{j+1}-t>0$ (this holds since the job $j+1$ appears on the shared processor in $\cS$) and $w_{j}\geq w_{j+1}$, we have that $\tct{\cS'}-\tct{\cS}\geq 0$.
Note that if $w_j$ is strictly greater than $w_{j+1}$, then we obtain the desired contradiction with the optimality of $\cS$.
However, if $w_j=w_{j+1}$, or in other words $\tct{\cS}=\tct{\cS'}$, then we need to use different arguments to arrive at a contradiction.

\medskip
To that end denote
\[q=\sum_{i'=i+2}^k \frac{w_{\pi(i')}}{2^{i'-i-1}}.\]
We show that
\begin{equation} \label{eq:wjq}
w_{j+1}\geq q.
\end{equation}
Suppose for a contradiction that $w_{j+1}<q$.
To obtain a contradiction with this assumption we will convert $\cS'$ into a schedule $\cS''$  with strictly greater total weighted overlap, which will contradict the optimality of the original schedule $\cS$.
This conversion is described in the next two paragraphs and depicted in Figure~\ref{fig:eliminate}.
\begin{figure}[htb]
\begin{center}
\includegraphics[scale=0.8]{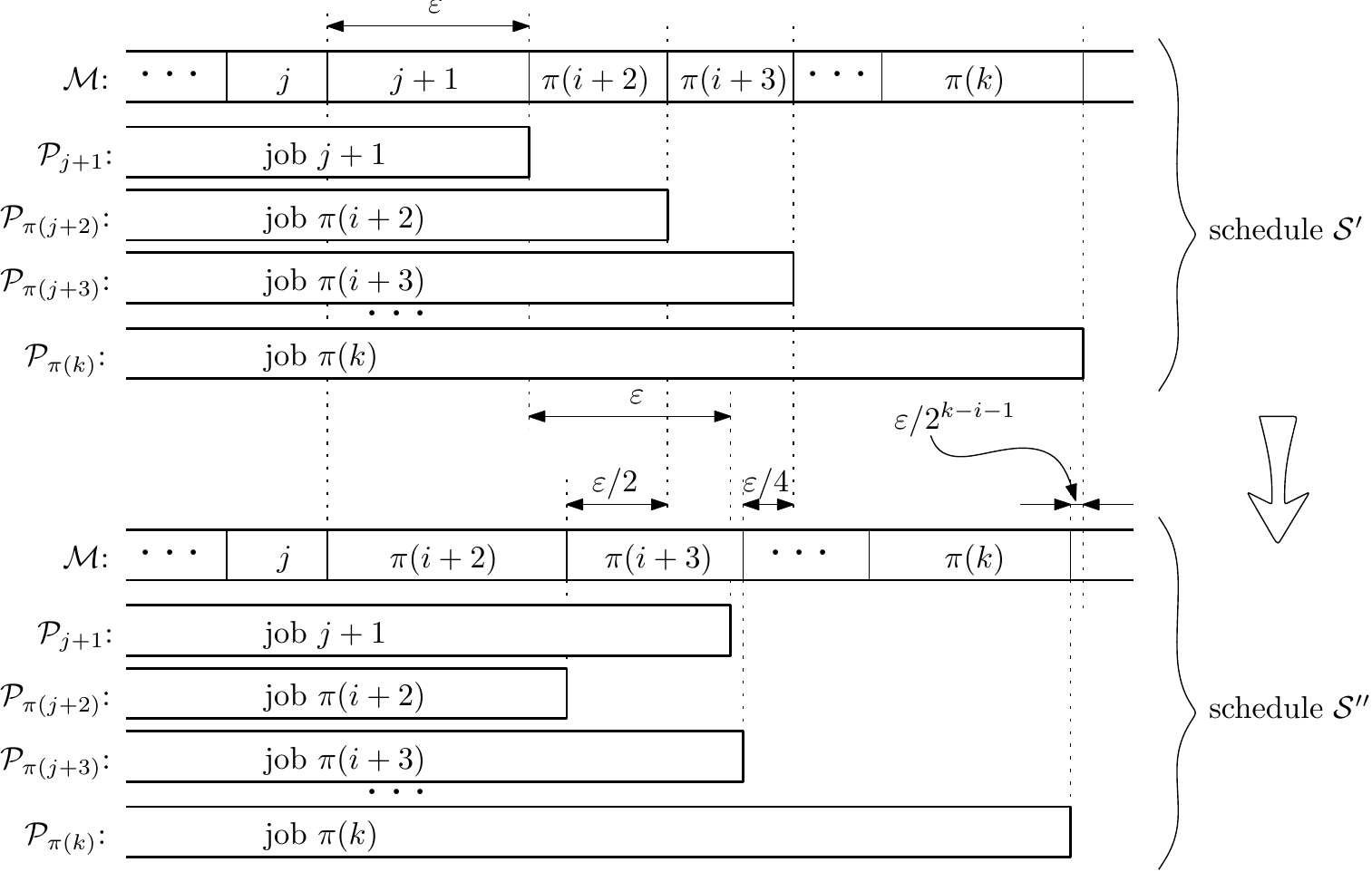}
\caption{Transition from $\cS'$ to $\cS''$ in proof of Lemma~\ref{lem:antithetical} when proving \eqref{eq:wjq}}
\label{fig:eliminate}
\end{center}
\end{figure}
Let
\[\varepsilon=\complTime{\cS'}{j+1}{\Mshared}-\startTime{\cS'}{j+1}{\Mshared}.\]
Observe that $\varepsilon>0$ by~\eqref{eq:lldiff}.
By assumption $\cS$ is synchronized and $\cS$ and $\cS'$ are identical on $\Mshared$ in time intervals $[\complTime{\cS'}{j+1}{\Mshared},+\infty)$. Thus 
\[\complTime{\cS'}{i'}{\Mshared}=\complTime{\cS}{i'}{\Mshared}=\complTime{\cS}{i'}{\Mpriv}=\complTime{\cS'}{i'}{\Mpriv}\]
for each job $i'$ that appears on $\Mshared$ and completes in $\cS'$ later than the job $j+1$.
The schedule $\cS''$ is defined as follows.
Let, $\cS''$ and $\cS'$ be identical on $\Mshared$ in time interval
\[\left[0,\complTime{\cS'}{j+1}{\Mshared}-\varepsilon\right)=\left[0,\startTime{\cS'}{j+1}{\Mshared}\right).\]
Then, the job $j+1$ is not present on $\Mshared$ in $\cS'$. It executes only on $\Mpriv_{j+1}$ in $\cS''$.
Finally, for each job $\pi(i')$, $i'\in\{i+2,\ldots,k\}$, we set:
\[\startTime{\cS''}{\pi(i')}{\Mshared}=\startTime{\cS'}{\pi(i')}{\Mshared}-\frac{\varepsilon}{2^{i'-i-2}},\]
\[\complTime{\cS''}{\pi(i')}{\Mshared}=\complTime{\cS'}{\pi(i')}{\Mshared}-\frac{\varepsilon}{2^{i'-i-1}},\]
\[\complTime{\cS''}{\pi(i')}{\Mpriv}=\complTime{\cS'}{\pi(i')}{\Mpriv}-\frac{\varepsilon}{2^{i'-i-1}}.\]
Both $\cS''$ and $\cS'$ are the same on $\Mpriv_{i'}$ for each $i'\in\jobs\setminus\{j+1,\pi(i+2),\ldots,\pi(k)\}$, i.e., on each processor not specified by the formulas above.

Clearly, $\cS''$ is feasible and synchronized.
To compare its total weighted overlap to that of $\cS'$, note that on the one hand the value of $\cS''$  decreases by $\varepsilon w_{j+1}$ in comparison to $\cS'$ since the job $j+1$ does not appear on $\Mshared$ in $\cS''$, on the other hand it increases since the subintervals with the jobs that follow $j+1$ on $\Mshared$ get longer due to $j+1$ disappearance from $\Mshared$. Hence
\[\tct{\cS''}=\tct{\cS'}-\varepsilon w_{j}+\varepsilon\sum_{i'=i+2}^k \frac{w_{i'}}{2^{i'-i-1}}=\tct{\cS'}+\varepsilon(q-w_{j})>0\]
because $\varepsilon>0$ and $w_{j}<q$ by assumption.
Thus, we obtain a contradiction with the optimality of $\cS$ (recall that $\tct{\cS}\leq\tct{\cS'}$).

\medskip
Having proved~\eqref{eq:wjq}, we complete the proof of case~\eqref{eq:contradictp} by performing another transformation of the schedule $\cS'$ to a new schedule $\cS''$.
We intend this transformation to also apply to the case of~\eqref{eq:contradictA}. Hence we will conduct the remaining part of the proof in such a way that we complete the proof of case~\eqref{eq:contradictp} and carry out a complete proof of case~\eqref{eq:contradictA} all at one time.
Since, we will obtain a contradiction in both cases, the proof of both~\eqref{eq:AisJ} and~\eqref{eq:pi} will be completed.
Thus, what we need is to define the $\cS'$ and the job $j+1$ in order to include the case~\eqref{eq:contradictA} into our proof:
for~\eqref{eq:contradictA} we take $\cS':=\cS$, $j=\pi(i+1)$ and $j+1=\bar{j}$.
(Recall that $\bar{j}$ is defined in~\eqref{eq:contradictA} to be a job that does not appear on $\Mshared$ and has processing time that is greater or equal to that of $\pi(i+1)$.)
It will follow from the construction below that the transformation works for both $\cS'$, $j$ and $j+1$ used in the earlier part of the proof of case~\eqref{eq:contradictp} and  for the new $\cS$, $\pi(i)$ and $\bar{j}$ in~\eqref{eq:contradictA}.
Observe that, informally speaking, the only difference in both cases is that in case~\eqref{eq:contradictp} the job $j+1$ is present on $\Mshared$ and completes on $\Mpriv$ later than on $\Mshared$ while in case~\eqref{eq:contradictA} the job $j+1$ is not present on $\Mshared$.
\begin{figure}[htb]
\begin{center}
\includegraphics[scale=0.8]{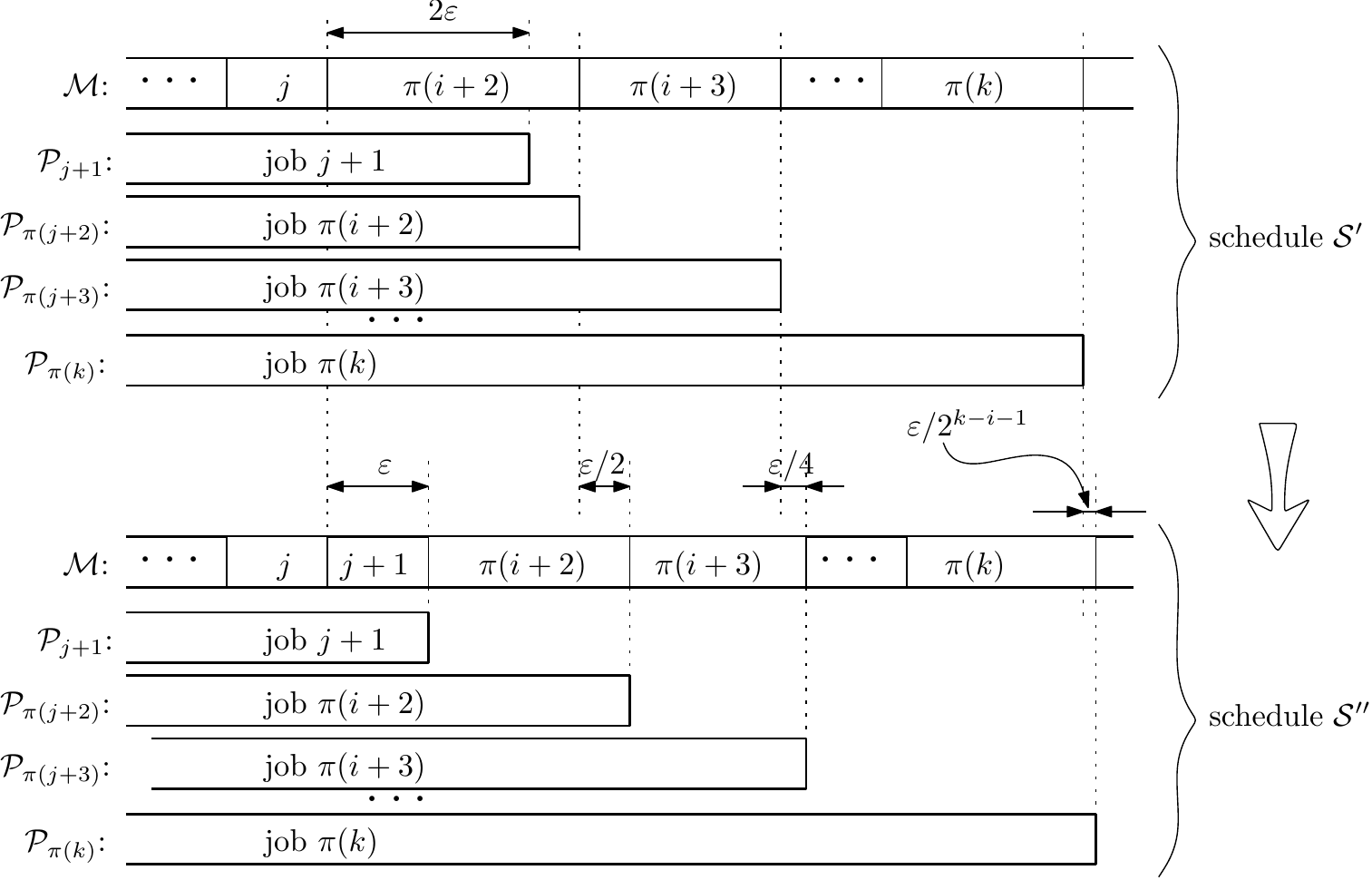}
\caption{Transition from $\cS'$ to $\cS''$ in proof of Lemma~\ref{lem:antithetical} for the completion of the proof of~\eqref{eq:AisJ} and~\eqref{eq:pi}; in this particular example we have that the job $j+1$ does not appear on $\Mshared$, thus presenting the case when we take $j+1=\bar{j}$}
\label{fig:introduce}
\end{center}
\end{figure}
We have in particular
\begin{equation} \label{eq:complLarge}
\complTime{\cS'}{j+1}{\Mpriv}>\complTime{\cS'}{j+1}{\Mshared}.
\end{equation}

Let
\[\varepsilon = \frac{1}{2}\left(\complTime{\cS'}{j+1}{\Mpriv}-\startTime{\cS'}{\pi(i+2)}{\Mshared}\right).\]
We now obtain a schedule $\cS''$ as follows.
First, $\cS''$ and $\cS'$ are identical on $\Mshared$ in time interval $[0,\complTime{\cS'}{j}{\Mshared})$,
\[\startTime{\cS''}{\pi(i')}{\Mshared}=\startTime{\cS'}{\pi(i')}{\Mshared}+\frac{\varepsilon}{2^{i'-i-2}}\]
and
\[\complTime{\cS''}{\pi(i')}{\Mshared}=\complTime{\cS''}{\pi(i')}{\Mpriv}=\complTime{\cS'}{\pi(i')}{\Mshared}+\frac{\varepsilon}{2^{i'-i-1}}\]
for each $i'\in\{i+2,\ldots,k\}$.
Informally, the start of the $\pi(i')$-th job is postponed on $\Mshared$ by $\varepsilon/2^{i'-i-2}$ and its completion time is set up in such a way that this job completes on both $\Mshared$ and $\Mpriv$ at the same time.
Then, $\startTime{\cS''}{j+1}{\Mshared}=\complTime{\cS''}{j}{\Mshared}$, $\complTime{\cS''}{j+1}{\Mshared}=\startTime{\cS''}{\pi(i+2)}{\Mshared}$ and the remaining part of $j+1$ executes on $\Mpriv_{j+1}$.
Finally, the execution of the remaining jobs (those for which completion times on their private processors have not been defined above) on their private processors remains the same in $\cS''$ as in $\cS'$.

We argue that the schedule $\cS''$ is feasible.
Clearly, the new time intervals in which job $j+1$ executes on $\Mshared$ and $\Mpriv_{j+1}$ are correct.
Moreover, $j+1$ completes on both processors at the same time.
For the jobs that follow $j+1$ on $\Mshared$, each of them also completes at the same time on $\Mshared$ and its private processor. Therefore these jobs are executed correctly because by the choice of $i$, they are ordered on $\Mshared$ (in $\cS'$ and also on $\cS''$) with non-decreasing values of their processing times:
\[p_{j+1}\leq p_{\pi(i+2)}\leq\cdots\leq p_{\pi(k)}.\]
This implies that $\cS''$ is feasible and synchronized.

From the construction we obtain that
\begin{equation} \label{eq:tctsComp}
\tct{\cS''}=\tct{\cS'}+\varepsilon w_{j+1}-\varepsilon\sum_{i'=i+2}^k \frac{w_{\pi(i')}}{2^{i'-i-1}}.
\end{equation}
Therefore, in case~\eqref{eq:contradictp}, we obtain by~\eqref{eq:wjq} that $\tct{\cS''}\geq\tct{\cS'}$.
In particular, if $\tct{\cS''}>\tct{\cS'}$, then we have immediately a contradiction with the optimality of $\cS$ since $\tct{\cS'}\geq\tct{\cS}$.
On the other hand, if $\tct{\cS''}=\tct{\cS'}$, then the contradiction comes from the selection of $\cS$ to be a schedule that minimizes the violation point $i$ (both $\cS''$ and $\cS$ have the same number of jobs present on the shared processor).

To consider case~\eqref{eq:contradictA}, observe that since $i$ is the violation point of $\cS$, $w_{\pi(i+2)}\geq w_{\pi(i')}$ for each $i'\in\{i+2,\ldots,k\}$.
Therefore,
\[\sum_{i'=i+2}^k \frac{w_{\pi(i')}}{2^{i'-i-1}} \leq w_{\pi(i+2)}\sum_{i'=i+2}^k \frac{1}{2^{i'-i-1}}< w_{\pi(i+2)}.\]
Hence by~\eqref{eq:tctsComp}, $\tct{\cS''}>\tct{\cS'}$ because $w_{j+1}\geq w_{\pi(i+2)}$.
The latter inequality comes from the fact that the instance is antithetical and $p_{j+1}\leq p_{\pi(i+2)}$, which comes from the maximality of $i$ with respect to the conditions~\eqref{eq:contradictp} and~\eqref{eq:contradictA}.
Thus, we again have a contradiction with the choice of $\cS$, which completes the proof of the lemma.
\end{proof}

\section{A 1/2-approximation Algorithm} \label{sec:approx}
\newcommand{\cSkey}{\cS_{\textup{key}}}

Let $0=q_0<q_{1}<\cdots<q_{\ell}$ and $u_{1},\ldots,u_{\ell}\geq 0$ for some $\ell \geq 1$.
An \textit{envelope} \ for $q_{1},\ldots,q_{\ell}$ and $u_{1},\ldots,u_{\ell}$ is a step-function of non-negative $x$ defined as follows
$$e(q_{1},\ldots,q_{\ell},u_{1},\ldots,u_{\ell},x)=\left\{ 
\begin{array}{cc}
u_{1} & \textup{if }q_{0}\leq x\leq q_{1} \\ 
u_{2} & \textup{if }q_{1}<x\leq q_{2} \\  
&\ldots  \\
u_{\ell} & \textup{if }q_{k-1}<x\leq q_{\ell} \\ 
0 & \textup{if }q_{\ell}<x.%
\end{array}%
\right.$$
The \textit{area} of the envelope $e$ is
\[
\sum_{i=1}^{\ell}u_{i}(q_{i}-q_{i-1}).
\]

\medskip
Let $\jobs$ be a set of jobs.
%Let $\jobs$ be the set of $n$ jobs where the $i$-th jobs has processing time $p_i$ and weight $w_i$, $i\in\{1,\ldots,n\}$.
Without loss of generality we assume $p_{1}\leq \cdots\leq p_{n}$, for any tie we assume the jobs that are tied are ordered in ascending order of their weights, i.e. the heaviest tied job comes last in the tie.
A sequence of jobs $i_1,\ldots,i_{\ell}$ for some $\ell\geq 1$, where $1\leq i_{1}<\cdots<i_{\ell}\leq n$, is called a \emph{key sequence} for $\jobs$ if it satisfies the following conditions:
\begin{enumerate} [label={\normalfont{(\roman*)}}]
\item\label{it:key0} $i_{\ell}=n$,
\item\label{it:key1} $w_{i_{1}}>\cdots>w_{i_{\ell}}$,
\item\label{it:key2} $w_{k}\leq w_{i_{j}}$ for each $k\in I_{i_{j}}=\{i_{j-1}+1,\ldots,i_{j}\}$ and $j\in\{1,\ldots,\ell \}$, where $i_0=0$.
\end{enumerate}
Clearly $p_{i_{1}}<\cdots<p_{i_{\ell}}$, thus $u(p_{i_{1}},\ldots,p_{i_{\ell}},w_{i_{1}},\ldots,w_{i_{\ell}},x)$ is an envelope; we refer to it as the \textit{upper envelope} for $\jobs$.
Let $u^{\ast }$ be the area of the upper envelope for $\jobs$.

Note that the key sequence always exists.
This follows from the fact that it can be constructed `greedily' by starting with picking the last job of the sequence (see~\ref{it:key0}) and then iteratively selecting the predecessor of the previously selected job so that the predecessor has strictly bigger weight (see~\ref{it:key1}) and satisfies the condition \ref{it:key2}.
Also, the key sequence is unique by the same argument.

We have the following simple observation.
\begin{claim}
For each $k \in\{1,\ldots,\ell\}$, $w_{i_{k}}=\max\{w_{j}\st i_{k}\leq j\leq n\}$. 
\qed
\end{claim}

The key sequence for $\jobs$ defines a synchronized schedule $\cSkey$ for $\jobs$ with the set of jobs executed on the shared processor being $\jobs_{\textup{key}}=\{i_{1},\ldots,i_{\ell}\}$ and the permutation of the jobs on the processor being $\pi (j)=i_{j}$ for $j\in\{1,\ldots,\ell\}$.
The jobs in $\jobs\setminus \jobs_{\textup{key}}$ are executed on their private processors only.
Following our notation introduced in Section \ref{sec:preliminaries}, we get $\cSkey=(i_1,\ldots,i_{\ell})$.
We have the following lemma.

\begin{lemma} \label{lem:envelopeUpper}
For the schedule $\cSkey$ it holds $2\tct{\cSkey}\geq u^{\ast }$.
\end{lemma}
\begin{proof}
We argue that for each $k\in\{1,\ldots,\ell\}$,
\begin{equation} \label{eq:times2}
w_{i_{k}}\left(\complTime{\cSkey}{i_{k}}{\Mshared}-\startTime{\cSkey}{i_{k}}{\Mshared}\right) \geq w_{i_{k}}\frac{p_{i_{k}} - p_{i_{k-1}}}{2},
\end{equation}
where $p_{i_0}=0$.
Note that $\startTime{\cSkey}{i_{k}}{\Mshared}\leq p_{i_{k-1}}$ for each $k\in\{1,\ldots,\ell\}$.
Thus,
\[\complTime{\cSkey}{i_{k}}{\Mshared}-\startTime{\cSkey}{i_{k}}{\Mshared}= \frac{\startTime{\cSkey}{i_{k}}{\Mshared}+p_{i_k}}{2} -\startTime{\cSkey}{i_{k}}{\Mshared} \geq \frac{p_{i_k}-p_{i_{k-1}}}{2}\]
for each $k\in\{1,\ldots,\ell\}$, which proves \eqref{eq:times2}.

By \eqref{eq:times2},
\[\tct{\cSkey} = \sum_{k=1}^{\ell} w_{i_{k}}\left(\complTime{\cSkey}{i_{k}}{\Mshared}-\startTime{\cSkey}{i_{k}}{\Mshared}\right) \geq \sum_{k=1}^{\ell} w_{i_{k}}\frac{p_{i_{k}} - p_{i_{k-1}}}{2} = \frac{u^{\ast}}{2}.\]
\end{proof}

We now prove that the area $u^{\ast }$ of the upper envelope for $\jobs$ is an upper bound on the value of optimal solution for $\jobs$.
By Theorem~\ref{thm:synchronized}, there exists an optimal synchronized schedule $\cS_{\textup{opt}}$ for $\jobs$.
Let $\jobs_{\textup{opt}}\subseteq\jobs$ be the set of jobs that appear on the shared processor in $\cS_{\textup{opt}}$ and let $\pi$ be a permutation of jobs in $\jobs_{\textup{opt}}$ in $\cS_{\textup{opt}}$.
Thus, we have $\cS_{\textup{opt}}=(\pi(1),\ldots,\pi(|\jobs_{\textup{opt}}|))$.
It holds $0<\complTime{\cS_{\textup{opt}}}{\pi(1)}{\Mshared} < \cdots < \complTime{\cS_{\textup{opt}}}{\pi(|\jobs_{\textup{opt}}|)}{\Mshared} <p_{n}$ and therefore
\[
e\left(\complTime{\cS_{\textup{opt}}}{\pi(1)}{\Mshared},\ldots,\complTime{\cS_{\textup{opt}}}{\pi(|\jobs_{\textup{opt}}|)}{\Mshared},w_{\pi (1)},\ldots,w_{\pi
(|\jobs_{\textup{opt}}|)},x\right).
\]
is an envelope for $\complTime{\cS_{\textup{opt}}}{\pi(1)}{\Mshared}, \ldots, \complTime{\cS_{\textup{opt}}}{\pi(|\jobs_{\textup{opt}}|)}{\Mshared}$ and $w_{\pi
(1)},\ldots,w_{\pi (|\jobs_{\textup{opt}}|)}$.
Let the area of this envelope be $e^{\ast }$.
We have the following  key result.

\begin{lemma} \label{lem:envelopeLower}
It holds $e^{\ast }\leq u^{\ast }$.
\end{lemma}
\begin{proof}
Observe that for each index $j\in\{1,\ldots,|\jobs_{\textup{opt}}|\}$ there exists $\tau(j)\in\{1,\ldots,n\}$ such that
\begin{equation} \label{eq:eLpp}
\complTime{\cS_{\textup{opt}}}{\pi(j)}{\Mshared} \leq p_{i_{\tau(j)}}.
\end{equation}
This follows from condition~\ref{it:key0} in definition of key sequence.
If, for a given $j$, there are several jobs $\tau(j)$ that satisfy the above, then take $\tau(j)$ to be the minimum one.

We argue that
\begin{equation} \label{eq:eLww}
w_{\pi(j)}\leq w_{i_{\tau(j)}}
\end{equation}
for each $j\in\{1,\ldots,|\jobs_{\textup{opt}}|\}$.
Suppose for a contradiction that \eqref{eq:eLww} does not hold.
We consider two cases.
In the first case let
\[p_{\pi(j)}\leq p_{i_{\tau(j)}}.\]
By condition~\ref{it:key2} in definition of the key sequence and the minimality of $\tau(j)$, $\pi(j)$ does not belong to the key sequence.
But then, $w_{\pi(j)} > w_{i_{\tau(j)}}$ implies that there is $t$ such that $p_{\pi(j)}<p_{i_t}<p_{i_{\tau(j)}}$, which contradicts the choice of $\tau(j)$.
In the second case let
\[p_{\pi(j)} > p_{i_{\tau(j)}}.\]
Take the minimum index $t$ such that $p_{i_t}\geq p_{\pi(j)}$.
By condition~\ref{it:key2} in definition of the key sequence, $w_{i_t}\geq w_{\pi(j)}$.
Since $w_{\pi(j)} > w_{i_{\tau(j)}}$, condition~\ref{it:key1} in definition of the key sequence implies that $i_{\tau(j)}$ does not belong to the key sequence --- a contradiction.
This completes the proof of~\eqref{eq:eLww}.

Since the upper envelope for $\jobs$ is non-increasing function in $x$, we obtain by~\eqref{eq:eLpp} and~\eqref{eq:eLww} that
\[e\left(\complTime{\cS_{\textup{opt}}}{\pi(1)}{\Mshared},\ldots,\complTime{\cS_{\textup{opt}}}{\pi(|\jobs_{\textup{opt}}|)}{\Mshared},w_{\pi (1)},\ldots,w_{\pi
(|\jobs_{\textup{opt}}|)},x\right) \leq u\left(p_{i_{1}},\ldots,p_{i_{\tau}},w_{i_{1}},\ldots,w_{i_{\tau}},x\right)\]
for each $x\geq 0$, which completes the proof.
\end{proof}

Since $\tct{\cS_{\textup{opt}}}=e^{\ast}$ and $\tct{\cSkey}\leq\tct{\cS_{\textup{opt}}}$, Lemmas~\ref{lem:envelopeUpper} and~\ref{lem:envelopeLower} give the following.
\begin{corollary} \label{cor:key}
It holds $e^{\ast }/2\leq \tct{\cSkey} \leq e^{\ast }$.
\qed
\end{corollary}

\begin{theorem}
The key sequence for $\jobs$ provides a $1/2$-approximation solution to the problem~$\probShort$.
This sequence can be found in time $O(n\log n)$ for any set of jobs $\jobs$, where $n=|\jobs|$.
Moreover, the bound of $1/2$ is tight, i.e., for each $\varepsilon>0$ there exists a problem instance such that $\tct{\cSkey}<\left(\frac{1}{2}+\varepsilon\right)\tct{\cS_{\textup{opt}}}$.
\end{theorem}
\begin{proof}
The fact that the key sequence is a $1/2$-approximation of the optimal solution follows from Corollary~\ref{cor:key}.
The key sequence can be constructed directly from the definition and sorting the jobs in $\jobs$ according to their processing times determines the $O(n\log n)$ running time.

To close we show that the 1/2 bound for the key sequences is tight.
Take $\jobs$ to contain $n$ jobs, each of the same weight $w>0$ and the same length $p>0$.
The key sequence consists of one job and therefore the value of the corresponding schedule $\cSkey$ is $\tct{\cSkey}=wp/2$.
Take a schedule $\cS$ that places all jobs in $\jobs$ on the shared processor.
We have $\tct{\cS}=wp(1-1/2^{n})$.
For $\varepsilon>0$ take $n=\lceil\log(1/\varepsilon)\rceil$.
If $\cS_{\textup{opt}}$ is an optimal schedule, then
\[\frac{\tct{\cSkey}}{\tct{\cS_{\textup{opt}}}} \leq \frac{\tct{\cSkey}}{\tct{\cS}} = \frac{1}{2}-\frac{1}{2^{n+1}}\leq \frac{1}{2}-\varepsilon.\]
\end{proof}

\section{Open problems and further research} \label{sec:summary}
The complexity status of $\probShort$ remains open.
The generalized problem with multiple shared processors is strongly NP-hard \cite{DK16} when the number of shared processors is a part of the input.
However, it remains open whether the generalized problem with fixed number of processors is NP-hard (i.e., whether the multi-processor problem is FPT with respect to the number of processors).
This complexity result and the open complexity questions
clearly underline the difficulty in finding efficient optimization algorithms for the shared processor scheduling problem. The development of an efficient branch-and-bound 
algorithm for the problem remains unexplored so far. The 1/2-approximation algorithm along with the structural properties of optimal schedules  presented in this paper and in \cite{DK16}
may prove useful building blocks of such an algorithm.

\section*{Acknowledgements}
This research has been supported by the Natural Sciences and Engineering Research
Council of Canada (NSERC) Grant OPG0105675 and by Polish National Science Center under contract DEC-2011/02/A/ST6/00201.

\bibliographystyle{plain}
\bibliography{references}

\end{document}